\documentclass[11pt,reqno]{amsart}
\usepackage{amssymb}
\usepackage{amsfonts}
\usepackage{amsmath}
\usepackage{stmaryrd}
\usepackage{physics}
\usepackage{braket}
\usepackage{dsfont}
\usepackage{bbold}
\usepackage{graphicx}
\usepackage{relsize}
\usepackage[table,xcdraw]{xcolor}
\usepackage{enumerate}
\usepackage[pagebackref, colorlinks = true, linkcolor = blue, urlcolor  = blue, citecolor = red]{hyperref}
\usepackage[margin=1in]{geometry}
\usepackage{enumitem}
\usepackage{mathtools}
\usepackage{graphbox}
\usepackage{comment}
\usepackage{float}
\usepackage[font=scriptsize]{caption}

\renewcommand{\epsilon}{\varepsilon}
\renewcommand{\phi}{\varphi}
\newcommand{\overbar}[1]{\mkern 1.5mu\overline{\mkern-1.5mu#1\mkern-1.5mu}\mkern 1.5mu}

\newcommand{\CDUC}{\mathsf{CDUC}}
\newcommand{\DUC}{\mathsf{DUC}}
\newcommand{\DOC}{\mathsf{DOC}}
\newcommand{\EWP}{\mathsf{EWP}}
\newcommand{\PSD}{\mathsf{PSD}}
\newcommand{\PCP}{\mathsf{PCP}}
\newcommand{\TCP}{\mathsf{TCP}}
\newcommand{\MLDUI}[1]{\mathcal{M}_{#1}(\mathbb{C})^{\times 2}_{\mathbb{C}^{#1}}}
\newcommand{\MLDOI}[1]{\mathcal{M}_{#1}(\mathbb{C})^{\times 3}_{\mathbb{C}^{#1}}}
\newcommand{\M}[1]{\mathcal{M}_{#1}(\mathbb{C})}
\newcommand{\C}[1]{\mathbb{C}^{#1}}
\newcommand{\Mreal}[1]{\mathcal{M}_{#1}(\mathbb{R})}
\newcommand{\T}[1]{\mathcal{T}_{#1}(\mathbb{C})}

\newtheorem{theorem}{Theorem}[section]
\newtheorem{definition}[theorem]{Definition}
\newtheorem*{definition*}{Definition}
\newtheorem{proposition}[theorem]{Proposition}

\newtheorem{lemma}[theorem]{Lemma}
\newtheorem{remark}[theorem]{Remark}

\newtheorem{conjecture}[theorem]{Conjecture}
\newtheorem*{conjecture*}{Conjecture}

\theoremstyle{definition}

\definecolor{darkgreen}{rgb}{0,0.392,0}

\author{Satvik Singh}
\email{satviksingh2@gmail.com}
\address{\parbox{\linewidth}{Department of Applied Mathematics and Theoretical Physics,\\
University of Cambridge, Cambridge, United Kingdom.}}

\author{Ion Nechita}
\email{nechita@irsamc.ups-tlse.fr}
\address{Laboratoire de Physique Th\'eorique, Universit\'e de Toulouse, CNRS, UPS, France}

\title{The PPT$^2$ conjecture holds for all Choi-type maps}

\begin{document}
\begin{abstract}
    We prove that the PPT$^2$ conjecture holds for linear maps between matrix algebras which are covariant under the action of the diagonal unitary group. Many salient examples, like the Choi-type maps, depolarizing maps, dephasing maps, amplitude damping maps, and mixtures thereof, lie in this class. Our proof relies on a generalization of the matrix-theoretic notion of factor width for pairwise completely positive matrices, and a complete characterization in the case of factor width two. 
\end{abstract}

\maketitle

\tableofcontents

\section{Introduction}

Quantum channels model the evolution of quantum systems. Mathematically, they correspond to completely positive and trace preserving linear maps between matrix algebras. One important scenario in the rapidly developing field of quantum technologies is the distribution of quantum entanglement: which channels can be used to transmit a quantum particle which is entangled to another system in such a way that some entanglement in the total bipartite system is preserved? Quantum channels $\Phi$ which are useless for this task are dubbed \emph{entanglement breaking} \cite{horodecki2003entanglement}: the local application of $\Phi$ on any subsystem of a bipartite quantum state results in a separable (non-entangled) state. For qubit channels, this property is equivalent to a simpler \emph{PPT} property, which amounts to saying that both $\Phi$ and $\Phi \circ \top$ are completely positive, where $\top$ denotes the transposition map. The preceding equivalence ceases to hold for higher dimensional qudit channels. However, the \emph{PPT squared conjecture} posits that the composition of two arbitrary PPT linear maps must be entanglement breaking \cite{PPTsq, Christandl2018}. In particular, for a PPT channel, its composition with itself must be entanglement breaking. 

\begin{conjecture}
The composition of two arbitrary \emph{PPT} linear maps is entanglement breaking. 
\end{conjecture}

This conjecture is relevant for quantum information theory because it imposes constraints on the type of resources that can be distributed using quantum repeaters \cite{bauml2015limitations,christandl2017private}. Since it was first proposed in 2012 by M.~Christandl, the conjecture has garnered a lot of attention. It has been shown that the conjecture holds in the asymptotic limit, i.e., the distance between several iterates of a unital (or trace preserving) PPT map and the set of entanglement breaking maps tends to zero in the asymptotic limit \cite{Kennedy2017}. In \cite{Rahaman2018}, the authors proved that any unital PPT map becomes entanglement breaking after finitely many iterations of composition with itself; for other algebraic approaches, see \cite{Lami2015entanglebreak,hanson2020eventually,girard2020convex}. As noted above, the conjecture trivially holds for qubit maps. For the next dimension $d=3$, the conjecture has been proven independently in \cite{Christandl2018} and \cite{Chen2019}. For higher dimensions however, the validity of the conjecture still remains ambiguous \cite{collins2018ppt,jin2020investigation}. In infinite dimensional systems, the set of Gaussian maps has been shown to satisfy the conjecture \cite{Christandl2018}.

The main result of the current paper is the proof of this conjecture in the case when the two PPT maps are covariant with respect to the action of the \emph{diagonal unitary group}. More precisely, we consider the class of linear maps which are (conjugate-)invariant under the action of diagonal unitaries (see Definition \ref{def:DUC-CDUC-DOC}): $\forall X \in \M{d} \text{ and } U \in \mathcal{DU}_d$, we have either 
$$\Phi(UXU^*) = U^*\Phi(X)U \qquad \text{ or } \qquad \Psi(UXU^*) = U\Psi(X)U^*.$$
These maps, dubbed respectively, \emph{Diagonal Unitary Covariant} (DUC) and \emph{Conjugate Diagonal Unitary Covariant} (CDUC), were studied at length in \cite{Singh2020diagonal}. A variety of physically relevant classes of quantum channels are of this kind, like the \emph{depolarizing} and \emph{transpose depolarizing} channels, \emph{amplitude damping} channels, \emph{Schur multipliers}, etc. Several important properties of (C)DUC maps, such as complete positivity and copositivity, entanglement breaking property, and the like, were analyzed in detail. In this work, we focus on the PPT$^2$ conjecture for these maps, proving a \emph{stronger} version of the conjecture. The main technical tool in our proof is the characterization of a subclass of entanglement breaking covariant maps, which is related to the matrix-theoretic concept of \emph{factor width} \cite{Boman2005factor} for the cone of \emph{pairwise completely positive matrices} \cite{johnston2019pairwise}, see Theorem \ref{theorem:PSD-2} and \ref{theorem:PCP-2}. The following is an informal statement of our main result, Theorem \ref{thm:PPT2-C-DUC}:

\begin{theorem}
The composition of two arbitrary (conjugate) diagonal unitary covariant \emph{PPT} maps corresponds to a pairwise completely positive matrix pair with factor width two; in particular, it is entanglement breaking.
\end{theorem}

The paper is organized as follows. In Section \ref{sec:prelim}, we provide several equivalent statements of the PPT$^2$ conjecture and also review some useful facts about diagonal unitary covariant maps; all of them are proven in \cite{Singh2020diagonal}. In Section \ref{sec:factor-width} we introduce the notion of factor width for pairwise and triplewise completely positive matrices. These tools are used in Section \ref{sec:PPTsq} to prove the PPT$^2$ conjecture for (C)DUC maps. Finally, in Section \ref{sec:open} we discuss some open problems and future directions for research. 

\section{Review of diagonal unitary covariant maps}\label{sec:prelim}
In this section, we will briefly review several key aspects from the theory of diagonal unitary and orthogonal covariant maps between matrix algebras. For a more comprehensive discussion and proofs of the results stated here, the readers are referred to our previous work \cite[Sections 6-9]{Singh2020diagonal}. 

Let us first quickly set up the basic notation. We use Dirac's \emph{bra-ket} notation for vectors $v\in \C{d}$ and their duals $v^*\in (\C{d})^*$ as \emph{kets} $\ket{v}$ and \emph{bras} $\bra{v}$, respectively. For $\ket{v},\ket{w}\in \C{d}$, the rank one matrix $vw^*$ is then represented as an \emph{outer-product} $\ketbra{v}{w}$. $\{ \ket{i}\}_{i=1}^d$ denotes the standard basis of $\C{d}$. We collect all $d\times d$ complex matrices into the set $\M{d}$. Within $\M{d}$, the cones of \emph{entrywise non-negative} and (hermitian) \emph{positive semi-definite} matrices are denoted by $\EWP_d$ and $\PSD_d$, respectively. We denote the adjoint (conjugate transpose) of $A\in \M{d}$ by $A^*$. Sets of pairs and triples of matrices in $\M{d}$ with equal diagonals are represented as follows
\begin{align}
    \MLDUI{d} &\coloneqq \{ (A,B) \in \M{d}\times \M{d} \, \big| \, \operatorname{diag}(A)=\operatorname{diag}(B) \} \label{eq:MLDUI}\\
    \MLDOI{d} &\coloneqq \{ (A,B,C) \in \M{d}\times \M{d}\times \M{d} \, \big| \, \operatorname{diag}(A)=\operatorname{diag}(B)=\operatorname{diag}(C)\} \label{eq:MLDOI}
\end{align}
The set of all linear maps $\Phi: \M{d}\rightarrow \M{d}$ is denoted by $\T{d}$. A map $\Phi\in \T{d}$ is called \emph{positive} if $\Phi(X)\in \PSD_d$ for all $X\in \PSD_d$. If $\operatorname{id}\otimes \Phi:\M{n}\otimes \M{d}\rightarrow \M{n}\otimes \M{d}$ is positive for all $n\in \mathbb{N}$, where $\operatorname{id}\in \T{n}$ is the identity map, then $\Phi\in \T{d}$ is called \emph{completely positive} (CP). Every CP map in $\T{d}$ admits a (non-unique) Kraus representation $\Phi(X) = \sum_{j=1}^k \mathfrak{A}_j X \mathfrak{A}_j^*$, where $\{\mathfrak{A}_j\}_{j=1}^d \subseteq \M{d}$ and $k\leq d^2$. If $\Phi\circ \top$ is completely positive, where $\top$ acts on $\M{d}$ as the matrix transposition (with respect to the standard basis in $\C{d}$), then $\Phi\in \T{d}$ is called \emph{completely copositive} (coCP). We say that $\Phi\in \T{d}$ is \emph{positive partial transpose} (PPT) if it is both CP and coCP. If, for all positive semi-definite $X\in \M{d}\otimes \M{d}$, $[\operatorname{id}\otimes \Phi] (X)$ lies in the convex hull of product matrices $A\otimes B$ with $A,B\in \PSD_d$, i.e. $[\operatorname{id}\otimes \Phi] (X)$ is \emph{separable}, we say that $\Phi\in \T{d}$ is \emph{entanglement-breaking}. By \emph{quantum channels}, we understand completely positive maps in $\T{d}$ which also preserve trace, i.e. $\operatorname{Tr}\Phi(X)=\operatorname{Tr}X, \,\, \forall X\in \M{d}$.

Now, in order to reformulate the PPT$^2$ conjecture in the language of bipartite matrices, we need to introduce the notion of locality for CP maps between tensor products of matrix algebras. For our purposes, it suffices to look at the tripartite setting. We say that a CP linear map $\Phi: \M{d}\otimes \M{d}\otimes \M{d} \rightarrow \M{d}\otimes \M{d}\otimes \M{d}$ is \emph{separable} if it can be expressed as a finite sum $\Phi = \sum_{i\in I} \Phi^1_i \otimes \Phi^2_i \otimes \Phi^3_i$, where $\Phi^j_i\in \T{d}$ are CP for all $i\in I$ and $j\in \{1,2,3\}$. Using the Kraus representation of CP maps in $\T{d}$, it is easy to see that any such separable operation itself admits a Kraus representation of the form $\Phi(X) = \sum_{j\in J} (\mathfrak{A}_j \otimes \mathfrak{B}_j \otimes \mathfrak{C}_j) X (\mathfrak{A}_j \otimes \mathfrak{B}_j \otimes \mathfrak{C}_j)^*$, where $\mathfrak{A}_j,\mathfrak{B}_j,\mathfrak{C}_j \in \M{d}$ for all $j \in J$. Also recall that a bipartite matrix $X\in \M{d}\otimes \M{d}$ is said to be \emph{positive under partial transpose} (PPT) if both $X$ and $X^\Gamma = [\operatorname{id}\otimes \top](X)$ are positive semi-definite. Equipped with the appropriate terminology, we are now prepared to state several equivalent formulations of the PPT$^2$ conjecture.

\begin{proposition} \label{prop:PPT2}
The following statements are equivalent:
\begin{enumerate}
    \item $\forall$ \emph{PPT} linear maps $\Phi_1,\Phi_2\in \T{d}$: $ \Phi_1\circ \Phi_2 \text{ is entanglement breaking}.$
    \item $\forall$ \emph{PPT} bipartite matrices $\rho, \sigma\in \M{d}\otimes \M{d}$: $$\operatorname{Tr}_{2,3} \{ (\rho\otimes \sigma) (\mathbb{I} \otimes \ketbra{e}{e} \otimes \mathbb{I} \} \text{ is separable}.$$
    \item $\forall$ \emph{PPT} bipartite matrices $\rho, \sigma \in \M{d}\otimes \M{d}$,
    $\forall$ tripartite separable \emph{CP} linear maps $\Lambda: \M{d}\otimes [\M{d}\otimes \M{d}]\otimes \M{d} \rightarrow \M{d}\otimes [\M{d}\otimes \M{d}]\otimes \M{d}$: 
    $$ \operatorname{Tr}_{2,3} \{ \Lambda (\rho\otimes \sigma) \} \text{ is separable}. $$
\end{enumerate}
Here, $\ket{e} = \sum_{i=1}^d \ket{i}\otimes \ket{i} \in \mathbb{C}^d \otimes \mathbb{C}^d$ is a maximally entangled vector, $\mathbb{I}\in\M{d}$ is the identity matrix, and $\operatorname{Tr}_{2,3}\{\cdot\}$ denotes partial trace over the middle two tensor factors. 
\end{proposition}

\begin{proof}
The equivalence of (1) and (2) can be readily established by using the Choi-Jamio{\l}kowski isomorphism $J:\T{d}\rightarrow \M{d}\otimes \M{d}$ defined as 
$$J(\Phi) = \sum_{i,j=1}^d \Phi(\ketbra{i}{j}) \otimes \ketbra{i}{j},$$ which identifies PPT and entanglement breaking linear maps in $\T{d}$ with PPT and separable matrices in $\M{d}\otimes \M{d}$, respectively. Since $J(\Phi_1 \circ \Phi_2) = \operatorname{Tr}_{2,3} \{ (J(\Phi_1)\otimes J(\Phi_2)) (\mathbb{I} \otimes \ketbra{e}{e} \otimes \mathbb{I} \}$, the equivalence of (1) and (2) becomes evident.

Let us now assume that (2) holds. Consider an arbitrary separable CP map $\Lambda$ as given in (3) with Kraus operators $\{\mathfrak{A}_j \otimes \mathfrak{C}_j \otimes \mathfrak{B}_j \}_{j\in J}$, where $\mathfrak{A}_j, \mathfrak{B}_j\in \M{d}$ and $\mathfrak{C}_j\in \M{d}\otimes \M{d}$ for all $j\in J$. Then, we can write 

$$\Lambda (\rho\otimes \sigma) = \sum_j (\mathbb{I}\otimes \mathfrak{C}_j \otimes \mathbb{I} ) (\rho_j\otimes \sigma_j ) (\mathbb{I}\otimes \mathfrak{C}_j \otimes \mathbb{I})^*,$$ 

where $\rho_j = (\mathfrak{A}_j \otimes \mathbb{I}) \rho (\mathfrak{A}_j \otimes \mathbb{I})^*$ and $\sigma_j = (\mathbb{I} \otimes \mathfrak{B}_j ) \sigma (\mathbb{I} \otimes \mathfrak{B}_j )^* $ are again PPT. Thus,

\begin{align*}
    \operatorname{Tr}_{2,3} \{\Lambda (\rho\otimes \sigma) \} &= \sum_j \operatorname{Tr}_{2,3} \{ (\rho_j\otimes \sigma_j) (\mathbb{I} \otimes \mathfrak{C}_j^* \mathfrak{C}_j  \otimes \mathbb{I})  \} \\
    &=  \sum_{j\in J} \sum_{i=1}^d \lambda_{ij} \operatorname{Tr}_{2,3} \{ (\rho_j \otimes \sigma_j) (\mathbb{I} \otimes \ketbra{c_{ij}}{c_{ij}} \otimes \mathbb{I}) \},
\end{align*}
where, for each $j\in J$, $\sum_{i=1}^d \lambda_{ij}\ketbra{c_{ij}}{c_{ij}}$ is the spectral decomposition of the positive semi-definite matrix $\mathfrak{C}_j^* \mathfrak{C}_j$. Now, by writing $\ket{c_{ij}}= (\mathfrak{C'}_{ij}\otimes \mathbb{I}) \ket{e}$ for $\mathfrak{C'}_{ij}\in \M{d}$ and redefining $\varrho_{ij} = (\mathbb{I} \otimes \mathfrak{C'}_{ij})^* \rho_{j} (\mathbb{I} \otimes \mathfrak{C'}_{ij})$ for all $i,j$ (which are all again PPT), we obtain: 
$$ \operatorname{Tr}_{2,3} \{\Lambda (\rho\otimes \sigma) \} = \sum_{j\in J} \sum_{i=1}^d \lambda_{ij} \operatorname{Tr}_{2,3} \{ (\varrho_{ij}\otimes \sigma_j) (\mathbb{I} \otimes \ketbra{e}{e} \otimes \mathbb{I})\},  $$
whose separability trivially follows from our assumption.

Finally, if we assume that (3) holds, then (2) follows by choosing $\Lambda$ to be defined by a single Kraus operator of the form $\mathbb{I}\otimes \ketbra{e}{e} \otimes \mathbb{I}$. 
\end{proof}

Let us take a second to interpret the above result. Assume that there are three spatially separated parties: Alice, Bob, and Charlie, such that Charlie shares bipartite states $\rho$ and $\sigma$ with Alice and Bob, respectively (see Figure~\ref{fig:swap}). The objective is to transfer any entanglement that Charlie shares with Alice and Bob separately (via $\rho$ and $\sigma$) to shared entanglement between Alice and Bob. We can think of this as a generalized entanglement `swapping' task. What the PPT$^2$ conjecture says in this context is that the above task is impossible if $\rho$ and $\sigma$ are PPT. No matter what local operations the parties might wish to perform on their subsystems, Proposition~\ref{prop:PPT2} guarantees that the resulting state shared by Alice and Bob is separable. This has drastic implications in the realm of quantum key distribution using repeater devices, where such entanglement swapping procedures are heavily employed, see \cite{christandl2017private, bauml2015limitations}.   

\begin{figure}[H]
    \centering
    \includegraphics[scale=1.3]{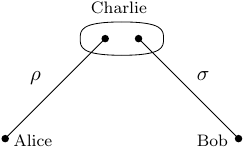}
    \caption{The setup for a generalized entanglement swapping task.}
    \label{fig:swap}
\end{figure}

We now define the different families of covariant maps in $\T{d}$. 

\begin{definition}\label{def:DUC-CDUC-DOC}
Let $\mathcal{DU}_d$ and $\mathcal{DO}_d$ denote the groups of diagonal unitary and diagonal orthogonal matrices in $\M{d}$, respectively. Then, a linear map $\Phi \in \mathcal{T}_d(\mathbb{C})$ is said to be
\begin{itemize}
    \item \emph{Diagonal Unitary Covariant (DUC)} if 
    $$ \forall X \in \M{d} \text{ and } U \in \mathcal{DU}_d: \qquad \Phi(UXU^*) = U^*\Phi(X)U,$$
    \item \emph{Conjugate Diagonal Unitary Covariant (CDUC)} if 
    $$ \forall X \in \M{d} \text{ and } U \in \mathcal{DU}_d: \qquad \Phi(UXU^*) = U\Phi(X)U^*,$$
    \item \emph{Diagonal Orthogonal Covariant (DOC)} if 
    $$ \forall X \in \M{d} \text{ and } O \in \mathcal{DO}_d: \qquad \Phi(OXO) = O\Phi(X)O.$$
\end{itemize}
\end{definition}

\begin{remark} \cite[Theorem 6.4]{Singh2020diagonal} \label{remark:DOC-LDOI}
Defintion~\ref{def:DUC-CDUC-DOC} can be reformulated in terms of bipartite Choi matrices having certain local diagonal unitary/orthogonal invariance properties:
\begin{itemize}
    \item $\Phi\in \T{d}$ is \emph{DUC}\phantom{C} $\iff (U \otimes  U) J(\Phi) (U^* \otimes  U^*) = J(\Phi) \qquad \forall U\in \mathcal{DU}_d$.
    \item $\Phi\in \T{d}$ is \emph{CDUC} $\iff (U \otimes  U^*) J(\Phi) (U^* \otimes  U) = J(\Phi) \qquad \forall U\in \mathcal{DU}_d$.
    \item $\Phi\in \T{d}$ is \emph{DOC}\phantom{C} $\iff (O \otimes  O) J(\Phi) (O \otimes O) = J(\Phi) \qquad \phantom{C}\forall O\in \mathcal{DO}_d$.
\end{itemize}
\end{remark}

The sets of DUC, CDUC and DOC maps in $\T{d}$ will be denoted by $\DUC_d, \CDUC_d$, and $\DOC_d$, respectively. Superscripts $i=1,2$ and $3$ will be used to distinguish between maps $\Phi^{(i)}$ in $\DUC_d, \CDUC_d$ and $\DOC_d$, respectively. Using the the structure of the invariant bipartite Choi matrices $J(\Phi^{(i)})$ from \cite[Proposition 2.3]{Singh2020diagonal}, one can parameterize the action of the corresponding covariant maps $\Phi^{(i)}$ on $\M{d}$ in terms of matrix triples $(A,B,C)\in \MLDOI{d}$ (Eq.~\eqref{eq:MLDOI}) as follows:
\begin{align}
    \Phi^{(1)}_{(A,B)}(X) &= \operatorname{diag}(A\ket{\operatorname{diag}X}) + \widetilde{B}\odot X^\top \label{eq:DUC-action} \\
    \Phi^{(2)}_{(A,B)}(X) &= \operatorname{diag}(A\ket{\operatorname{diag}X}) + \widetilde{B}\odot X \label{eq:CDUC-action} \\
    \Phi^{(3)}_{(A,B,C)}(X) &= \operatorname{diag}(A\ket{\operatorname{diag}X}) + \widetilde{B}\odot X + \widetilde{C}\odot X^\top \label{eq:DOC-action}
\end{align}
where $\widetilde{B}=B-\operatorname{diag}B, \, \widetilde{C}=C-\operatorname{diag}C$, and $\odot$ denotes the operation of \emph{Hadamard} (or \emph{entrywise}) product in $\M{d}$. In a quantum setting, where $X=\rho$ is a \emph{quantum state} ($\rho\in \PSD_d, \operatorname{Tr}\rho=1$) and $\Phi^{(i)}\in \T{d}$ are \email{quantum channels}, one can interpret the above actions by splitting them into two parts. The first part involves a \email{classical} diagonal mixing operation, which is nothing but a transformation on the space of probability distributions in $\mathbb{R}^d_{+}$: $\ket{\operatorname{diag}\rho} \mapsto A\ket{\operatorname{diag}\rho}$. The second part acts on the off-diagonal part of the input state by mixing the well-known actions of Schur Multipliers \cite[Chapters 3,8]{paulsen2002schur} and transposition maps in $\T{d}$. 

An important example of the above kind of maps is the \emph{Choi map} \cite{choi1975}, which was introduced in the `70s as the first example of a positive non-decomposable map: 
$$\Phi_{\mathsf{Choi}} : \M{3} \to \M{3}, \qquad \Phi_{\mathsf{Choi}}(X) = \left(   \begin{array}{ccc}
        X_{11}+X_{33} & -X_{12} & -X_{13}  \\
        -X_{21} & X_{11}+X_{22} & -X_{23} \\
        -X_{31} & -X_{32} & X_{22}+X_{33}
    \end{array}  \right).$$
It can be easily seen that the Choi map is a CDUC map $\Phi_{\mathsf{Choi}} = \Phi^{(2)}_{(A,B)}$, with 
$$A = \begin{pmatrix} 1 & 0 & 1 \\ 1 & 1 & 0 \\ 0 & 1 & 1 \end{pmatrix} \quad \text{ and } \quad B = \begin{pmatrix} \phantom{-}1 & -1 & -1 \\ -1 & \phantom{-}1 & -1 \\ -1 & -1 & \phantom{-}1 \end{pmatrix} = 2\mathbb{I}_3 - \mathbb{J}_3,$$
where $\mathbb J_3$ is the all-ones matrix. In fact, the action of all \emph{generalized Choi} maps in $\T{d}$ \cite{Ha2003choi,Chruscinski2007choi,Chruscinski2018choi} can be similarly parameterized by an arbitrary $A\in \EWP_d$ and $B=2\mathbb{I}_d-\mathbb{J}_d$, see \cite[Example 7.5]{Singh2020diagonal}. Besides Choi-type maps, the classes of (C)DUC and DOC maps contain many other important examples, like the \emph{depolarizing} and \emph{transpose depolarizing} maps, \emph{amplitude damping} maps, \emph{Schur multipliers}, etc. (see \cite[Section 7 and Table 2]{Singh2020diagonal} for a list of examples).

\begin{remark}
Note that the maps in $\DUC_d$ and $\CDUC_d$ are linked through composition by matrix transposition, i.e for $(A,B)\in \MLDUI{d}$, we have $\Phi^{(1)}_{(A,B)} = \Phi^{(2)}_{(A,B)}\circ \top$. We will shortly observe a reflection of this characteristic in the fact that the \emph{PPT} and entanglement-breaking properties of these maps entail an equivalent constraint on the corresponding matrix pairs $(A,B)\in \MLDUI{d}$. 
\end{remark}

\begin{remark}
For a matrix pair $(A,B)\in \MLDUI{d}$, we have $$  \Phi^{(1)}_{(A,B)} = \Phi^{(3)}_{(A,\operatorname{diag}A,B)} \qquad \Phi^{(2)}_{(A,B)} = \Phi^{(3)}_{(A,B,\operatorname{diag}A)} \qquad  $$
\end{remark}

Next, we introduce the cones of \emph{pairwise} and \emph{triplewise completely positive} matrices, which were first introduced in \cite{johnston2019pairwise} and \cite{nechita2021graphical}, respectively, as generalizations of the well-studied cone of completely positive matrices \cite{berman2003completely}. These will be used later in Propositions~\ref{prop:PPT-ent-DUC/CDUC} and \ref{prop:PPT-ent-DOC} to provide an equivalent description of the entanglement-breaking properties of our covariant families of maps. It is wise to point out that one should not confuse these notions with the earlier defined completely positive maps in $\T{d}$, which are different beasts altogether.

\begin{definition} \label{def:pcp}
A matrix pair $(A,B)\in \MLDUI{d}$ is said to be \emph{pairwise completely positive (PCP)} if there exist vectors $\{ \ket{v_n},\ket{w_n} \}_{n\in I}$ (for a finite index set $I$) such that 
\begin{equation*}
    A = \sum_{n\in I}\ketbra{v_n\odot \overbar{v_n}}{w_n\odot \overbar{w_n}}, \qquad B = \sum_{n\in I}\ketbra{v_n\odot w_n}{v_n\odot w_n}.
\end{equation*}
\end{definition}

\begin{definition} \label{def:tcp}
A matrix triple $(A,B,C)\in \MLDOI{d}$ is said to be \emph{triplewise completely positive (TCP)} if there exist vectors $\{ \ket{v_n},\ket{w_n} \}_{n\in I}$ (for a finite index set $I$) such that 
\begin{equation*}
    A = \sum_{n\in I}\ketbra{v_n\odot \overbar{v_n}}{w_n\odot \overbar{w_n}}, \qquad B = \sum_{n\in I}\ketbra{v_n\odot w_n}{v_n\odot w_n}, \qquad C = \sum_{n\in I}\ketbra{v_n\odot \overbar{w_n}}{v_n\odot \overbar{w_n}}. 
\end{equation*}
\end{definition}

The vectors $\{\ket{v_n},\ket{w_n}\}_{n\in I}$ above are said to form the PCP/TCP decomposition of the concerned matrix pair/triple. Notice that $(A,B,C)\in \TCP_d\implies (A,B),(A,C)\in \PCP_d$. It is easy to deduce that PCP and TCP matrices form closed convex cones, which we will denote by $\PCP_d$ and $\TCP_d$, respectively. For an extensive account of the convex structure of these cones, the readers should refer to \cite[Section 5]{Singh2020diagonal}. Several elementary properties of these cones are discussed in \cite[Sections 3,4]{johnston2019pairwise} and \cite[Appendix B]{nechita2021graphical}, respectively. We recall some important necessary conditions for membership in the $\PCP_d$ cone below. 

\begin{lemma} \label{lemma:pcp-necessary}
Let $(A,B)\in \PCP_d$. Then, $A\in \EWP_d$ and $B\in \PSD_d$. Moreover, the entrywise inequalities $A_{ij}A_{ji} \geq \vert B_{ij} \vert^2$ hold for all $i,j$.
\end{lemma}

Let us now describe the PPT and entanglement-breaking properties of the covariant maps in terms of constraints on the associated matrix pairs and triples.

\begin{proposition} \cite[Lemmas 6.11, 6.12]{Singh2020diagonal} \label{prop:PPT-ent-DUC/CDUC}
Let $(A,B)\in \MLDUI{d}$. Then, $\Phi^{(1)}_{(A,B)}$ is 
\begin{enumerate}
\item \emph{CP} $\iff A\in \EWP_d$, $B=B^*$ and $A_{ij}A_{ji}\geq |B_{ij}|^2 \,\, \forall i,j\iff \Phi^{(2)}_{(A,B)}$ is \emph{coCP}.
    \item \emph{coCP} $\iff A\in \EWP_d$ and $B\in \PSD_d\iff \Phi^{(2)}_{(A,B)}$ is \emph{CP}.
    \item \emph{PPT} $\iff A\in \EWP_d$, $B\in \PSD_d$ and $A_{ij}A_{ji}\geq |B_{ij}|^2 \,\, \forall i,j \iff  \Phi^{(2)}_{(A,B)}$ is \emph{PPT}.
    \item entanglement breaking $\iff (A,B)\in\PCP_d \iff\Phi^{(2)}_{(A,B)}$ is entanglement breaking.
\end{enumerate}
\end{proposition}

\begin{proposition} \cite[Lemma 6.13]{Singh2020diagonal} \label{prop:PPT-ent-DOC}
Let $(A,B,C)\in \MLDOI{d}$. Then, $\Phi^{(3)}_{(A,B,C)}$ is
\begin{enumerate}
\item \emph{CP} $\iff A \in \EWP_d$, $B \in \PSD_d$, $C=C^*$, and $A_{ij}A_{ji} \geq \vert C_{ij} \vert^2 \,\, \forall i,j$.
\item \emph{coCP} $\iff A \in \EWP_d$, $B=B^*$, $C \in \PSD_d$, and $A_{ij}A_{ji} \geq \vert B_{ij} \vert^2 \,\, \forall i,j$.
    \item \emph{PPT} $\iff A\in \EWP_d$, $B,C\in \PSD_d$ and $A_{ij}A_{ji}\geq \operatorname{max}\{|B_{ij}|^2,|C_{ij}|^2 \} \,\, \forall i,j$.
    \item entanglement breaking $\iff (A,B,C)\in \TCP_d$.
\end{enumerate}
\end{proposition}

\section{Factor widths}\label{sec:factor-width}

The concept of factor width was first formalized in \cite{Boman2005factor} for real (symmetric) positive semi-definite matrices, although the idea had been in operation before, particularly in the study of completely positive matrices \cite[Definition 2.4]{berman2003completely}. Recall that for $\ket{v}\in \C{d}$, its support is defined as $\operatorname{supp}\ket{v}\coloneqq \{ i\in [d] : v_i\neq 0\}$. We define $\sigma(v)$ as the size of $\operatorname{supp}\ket{v}$, that is the number of non-zero coordinates of $\ket v$. A real positive semi-definite matrix $A\in \Mreal{d}$ is said to have \emph{factor width} $k$ if there exists vectors $\{\ket{v_n}\}_{n\in I}\subset \mathbb{R}^d$ with $\sigma(v_n)\leq k$ for each $n$ such that $A$ admits the following rank one decomposition: $A = \sum_{n\in I}\ketbra{v_n}{v_n}$. Besides being heavily implemented in the analysis of completely positive matrices \cite[Section 4]{berman2003completely}, the concept of factor width has found several applications in the field of conic programming and optimization theory \cite{ahmadi2019dsos}. In this section, we will extend the notion of factor width to the cones of complex (hermitian) positive semi-definite matrices $\PSD_d$, pairwise completely positive matrices $\PCP_d$ and triplewise completely positive matrices $\TCP_d$. In particular, we will obtain a complete characterization of matrices with factor width $2$ in $\PSD_d$ and $\PCP_d$, which will later play an instrumental role in proving the validity of the PPT squared conjecture for covariant maps in $\DUC_d, \CDUC_d$. Without further delay, let us now delve straight into the definition of factor width for matrices in $\PSD_d$, $\PCP_d$ and $\TCP_d$.
\begin{definition} \label{def:psd-2}
    A matrix $B\in \PSD_d$ is said to have \emph{factor width} $k\in \mathbb{N}$ if it admits a rank one decomposition $B=\sum_{n\in I}\ketbra{v_n}{v_n}$, where $\{ \ket{v_n} \}_{n\in I} \subset \C{d} $ are such that $\sigma(v_n)\leq k$ for each $n\in I$.
\end{definition}

The notion of factor width for positive semidefinite matrices has been considered in \cite{ringbauer2018certification}, in relation to measures of coherence for density matrices. This notion has a straightforward generalization for PCP and TCP matrices.

\begin{definition} \label{def:pcp/tcp-2}
    A matrix pair $(A,B)\in \PCP_d$ (resp. triple $(A,B,C)\in \TCP_d$) is said to have \emph{factor width} $k\in \mathbb{N}$ if it admits a \emph{PCP} (resp. \emph{TCP}) decomposition with vectors $\{\ket{v_n},\ket{w_n} \}_{n\in I}\subset \C{d}$ such that $\sigma(v_n\odot w_n)\leq k$ for each $n\in I$.
\end{definition}
The cones of factor width $k$ matrices in $\PSD_d, \PCP_d$ and $\TCP_d$ will be denoted by $\PSD_d^k, \PCP_d^k$, and $\TCP_d^k$, respectively. It should be clear from the definitions that $(A,B,C)\in \TCP_d^k\implies (A,B),(A,C)\in \PCP_d^k \implies B,C\in \PSD_d^k$ and that the cones $\PSD_d^k, \PCP_d^k$, and $\TCP_d^k$ are stable by direct sums. The following sequences of inclusions are also trivial consequences of the definitions:
\begin{align}
    \PSD_d^1 \subset \PSD_d^2 \subset \dots \subset \PSD_d^d = \PSD_d \\
    \PCP_d^1 \subset \PCP_d^2 \subset \dots \subset \PCP_d^d = \PCP_d \\ 
    \TCP_d^1 \subset \TCP_d^2 \subset \dots \subset \TCP_d^d = \TCP_d
\end{align}

\begin{remark}\label{rem:strict-inclusion}
For all $k < d$, the inclusions
$$\PSD_d^k \subset \PSD_d, \qquad \PCP_d^k \subset \PCP_d, \qquad \TCP_d^k \subset \TCP_d$$
are strict. This can be seen by considering extremal rays of the cones (see \cite[Theorem 5.13]{Singh2020diagonal}) generated by vectors of full support.
\end{remark}

For $k=1$, it is evident that $\PSD_d^1$ is precisely the set of diagonal matrices in $\EWP_d$. It is equally easy to deduce that $\PCP_d^1$ (resp. $\TCP_d^1$) contain matrix pairs $(A,B)\in \PCP_d$ (resp. triples $(A,B,C)\in \TCP_d$) where $A\in \EWP_d$ and $B=\operatorname{diag}A$ (resp. $A\in \EWP_d$ and $B=C=\operatorname{diag}A$). 

The $k=2$ case is more interesting, and we must familiarize ourselves with some matrix-theoretic terminology before we begin to deal with it. Let us start with the definitions of the so-called \emph{scaled diagonally dominant} and M-matrices. In what follows, $\mathbb{I}_d$ denotes the identity matrix in $\M{d}$. For $B\in \M{d}$ and $k\in \mathbb{N}$, we define a hierarchy of comparison matrices entrywise as: 
\begin{equation}\label{eq:M-k}
    M_k(B)_{ij} = \begin{cases}
    \, k|B_{ij}|, \quad &\text{if } i = j\\
    -|B_{ij}|, \quad &\text{otherwise }
\end{cases}
\end{equation}
Notice that $M_1(B)$ coincides with the usual comparison matrix $M(B)$ for $B\in \M{d}$
\begin{definition}
    A matrix $B\in \M{d}$ is called \emph{diagonally dominant (DD)} if $|B_{ii}|\geq \sum_{j\neq i}|B_{ij}|$ and $|B_{ii}|\geq \sum_{j\neq i}|B_{ji}| \,\, \forall i$. For $B\in \M{d}$, if there exists a positive diagonal matrix $D$ such that $DBD$ is diagonally dominant, then $B$ is called \emph{scaled diagonally dominant (SDD)}.
\end{definition}

Note that, for a hermitian matrix $B$, the two conditions $|B_{ii}|\geq \sum_{j\neq i}|B_{ij}|$ and $|B_{ii}|\geq \sum_{j\neq i}|B_{ji}|$ are equivalent.
Also note that hermitian DD and SDD matrices are clearly positive semi-definite.
\begin{definition}
    A matrix $B=s\,\mathbb{I}_d-P$ with $s\geq 0$ and $P\in \EWP_d$ is called an \emph{M}-matrix if $s\geq \rho(P)$, where $\rho(P)$ denotes the spectral radius of $P$.
\end{definition}

From the above definition, it is easy to see that if $P\in \EWP_d$ is symmetric, then $B=s\,\mathbb{I}_d-P$ is an M-matrix if and only if it is positive semi-definite. Before proceeding further, let us equip ourselves with an important result from the Perron-Frobenius theory of non-negative matrices \cite[Chapter 8]{Horn2012matrix}. Recall that a matrix $B\in \M{d}$ is \emph{reducible} if it is permutationally similar to a block matrix in Eq.~\eqref{eq:reducible} (where $B_1,B_3$ are square matrices) and \emph{irreducible} otherwise.
\begin{equation} \label{eq:reducible}
    \left(
\begin{array}{ c c }
   B_1 & B_2 \\
   0 & B_3
\end{array}
\right)
\end{equation}
\begin{lemma}\label{lemma:Perron}
For an irreducible non-negative matrix $P\in \EWP_d$, the spectral radius $\rho(P)$ is an eigenvalue (called \emph{Perron} eigenvalue) of unit multiplicity with a positive eigenvector (called \emph{Perron} eigenvector) $\ket{p}\in \mathbb{R}^d_+$: $P\ket{p} = \rho(P)\ket{p}$.
\end{lemma}

With all the required tools now present in our arsenal, we begin to analyse the structure of the $\PSD_d^2$ and $\PCP_d^2$ cones. We start by showing that diagonal dominance is a sufficient condition to guarantee membership in these cones. It would be insightful to compare the following results with \cite[Theorem 2]{Barker1975diagonal} and \cite[Theorem 4.4]{johnston2019pairwise}.

\begin{lemma}\label{lemma:B-DD}
If $B\in \M{d}$ is a (hermitian) diagonally dominant matrix, then $B\in \PSD_d^2$. 
\end{lemma}
\begin{proof}
Let us use the symbol in Eq.~\eqref{eq:2x2} to define a $d\times d$ matrix which has zeros everywhere except in the $i,j$-prinicipal submatrix, where the entries are defined by the matrix present in the notation.
\begin{equation}\label{eq:2x2}
    \left(
\begin{array}{ c c }
   a & b \\
   c & d
\end{array}
\right)_{i,j\in [d]} \coloneqq a\ketbra{i}{i} + b\ketbra{i}{j} + c\ketbra{j}{i} + d\ketbra{j}{j} \in \M{d}
\end{equation}
Then, the following decomposition of a hermitian DD matrix shows that it has factor width $2$
\begin{equation}
    B = \sum_{1\leq i<j \leq d} \left(
\begin{array}{ c c }
   |B_{ij}| & B_{ij} \\
   B_{ji} & |B_{ji}|
\end{array}
\right)_{i,j\in [d]} + \operatorname{diag}\ket{b}
\end{equation}
where $\ket{b}\in \mathbb{R}^d$ is defined entrywise as $b_i=  B_{ii}-\sum_{j\neq i}|B_{ij}|\geq 0$.
\end{proof}

\begin{lemma}\label{lemma:AB-DD}
Let $(A,B)\in \MLDUI{d}$ be such that it satisfies the necessary conditions for membership in $\PCP_d$. (see Lemma~\ref{lemma:pcp-necessary}). Moreover, if $B$ is diagonally dominant, then $(A,B)\in \PCP_d^2$.
\end{lemma}
\begin{proof}
    Using the same notation as in Lemma~\ref{lemma:B-DD}, we can decompose the given pair $(A,B)$ as
\begin{equation}
    (A,B) = \sum_{1\leq i<j \leq d} \left( \left(
\begin{array}{ c c }
   |B_{ij}| & A_{ij} \\
   A_{ji} & |B_{ji}|
\end{array}
\right)_{i,j\in [d]} , \left(
\begin{array}{ c c }
   |B_{ij}| & B_{ij} \\
   B_{ji} & |B_{ji}|
\end{array}
\right)_{i,j\in [d]}     \right) + (\operatorname{diag}\ket{a},\operatorname{diag}\ket{b} )
\end{equation}
where $\ket{a}=\ket{b}\in \mathbb{R}^d$ is defined entrywise as $a_i=A_{ii}-\sum_{j\neq i}|B_{ij}|\geq 0$. Each pair in the above sum lies in $\PCP_d^2$, since it satisfies the lemma's hypothesis and is supported on a $2$-dimensional subspace, see \cite[Lemma B.11]{nechita2021graphical} or \cite[Theorem 4.1]{johnston2019pairwise}. Hence, $(A,B)\in \PCP_d^2$.
\end{proof}

With Lemmas~\ref{lemma:B-DD} and \ref{lemma:AB-DD} in place, we now proceed to obtain a complete characterization of the factor width $2$ cones $\PSD_d^2$ and $\PCP_d^2$. The following results can be interpreted as generalizations of a similar result for real symmetric matrices \cite[Theorem 9]{Boman2005factor}, where the set of factor width $2$ matrices has been shown to be equal to the set of scaled diagonally dominant matrices. A similar result has been obtained in \cite[Theorem 1]{ringbauer2018certification}, in the context of multilevel coherence of mixed quantum states. 

\begin{theorem} \label{theorem:PSD-2}
For a (hermitian) matrix $B\in \M{d}$, the following equivalences hold:
\begin{equation*}
    B\in \PSD_d^2 \!\iff\! M(B) \text{ is an \emph{M}-matrix} \!\iff\! M(B) \in \PSD_d \!\iff\!  B \text{ is scaled diagonally dominant}.
\end{equation*}
\end{theorem}
\begin{proof}
    Since $B$ is hermitian, $M(B)$ is symmetric, which implies that $M(B)$ is an M-matrix $\iff M(B)\in \PSD_d$. With this fact in the background, assume first that $B\in \PSD_d^2$ can be decomposed as $B=\sum_{n\in I}\ketbra{v_n}{v_n}$, where $\sigma(v_n)\leq 2$ for each $n$. Now, split the index set as $I=I_1\sqcup \bigsqcup_{i<j}I_{ij}$, where $ I_1=\{n\in I:\sigma(v_n)= 1 \}\text{ and }I_{ij}=\{n\in I : \operatorname{supp}\ket{v_n}=\{i,j\} \}$.
    By defining $B^{ij}=\sum_{n\in I_{ij}}\ketbra{v_n}{v_n}$ (so that $M(B^{ij})\in \PSD_d$ since $B$ is supported on a 2-dimensional subspace), we can deduce that $$M(B)=\sum_{n\in I_1}\ketbra{v_n}{v_n} + \sum_{i<j}M(B^{ij})\in \PSD_d. $$
    
    Now, let $B$ be a hermitian matrix such that $M(B)=s\,\mathbb{I}_d - P\in \PSD_d$ with $s\geq \rho(P)$. If $M(B)$ is reducible, it can be written as a direct sum of irreducible matrices. Hence, without loss of generality, we can assume that $M(B)$ is irreducible. Lemma~\ref{lemma:Perron} then provides us with a positive Perron eignenvector $\ket{p}\in \mathbb{R}^d_+$ such that $P\ket{p}=\rho(P)\ket{p}$. Define $D=\operatorname{diag}\ket{p}$ and let $\ket{e}\in \mathbb{R}^d$ be the all ones vectors, so that $DM(B)D\ket{e}=D(s-\rho(P))\ket{p}$ is entrywise non-negative, i.e $DM(B)D$ is diagonally dominant. The same $D$ then works to show that $B$ is scaled diagonally dominant.
    
    For the final implication, assume that there exists a positive diagonal matrix $D$ such that $DBD$ is diagonally dominant. Using Lemma~\ref{lemma:B-DD}, we then have $DBD\in \PSD_d^2 \implies B\in \PSD_d^2$.
\end{proof}
\begin{theorem} \label{theorem:PCP-2}
For $(A,B)\in \MLDUI{d}$ satisfying the necessary conditions for membership in $\PCP_d$ (see Lemma~\ref{lemma:pcp-necessary}), the following equivalence holds: $(A,B)\in \PCP_d^2 \iff B\in \PSD_d^2$.
\end{theorem}
\begin{proof}
    The forward implication is trivial to prove. For the reverse implication, let us assume that $B\in \PSD_d^2$. Theorem~\ref{theorem:PSD-2} then guarantees the existence of a positive diagonal matrix $D$ such that $DBD$ is (hermitian) diagonally dominant. One can then apply Lemma~\ref{lemma:AB-DD} on the pair $(DAD,DBD)$ to deduce that $(DAD,DBD)\in \PCP_d^2$, which clearly implies that $(A,B)\in \PCP_d^2$.
\end{proof}

Upon light scrutiny, it becomes clear that an analogue of Theorem~\ref{theorem:PCP-2} would not work for the $\TCP_d^2$ cone, due to the added complexity of the third matrix. In fact, counterexamples exist in the form of matrix triples $(A,B,C)\in \MLDOI{d}$ where, even though $(A,B)$ and $(A,C)$ separately satisfy the constraints of Lemma~\ref{lemma:AB-DD}, $(A,B,C)$ is not TCP, see \cite[Example 9.2]{Singh2020diagonal}.

Even in the real regime of symmetric positive semi-definite matrices, the cones of factor width $k\geq 3$ matrices have evaded tractability. Computing the factor width of a given real symmetric matrix may be NP-hard \cite{Boman2005factor}. Hence, it is reasonable to expect similar hardness in obtaining characterizations of the factor width $k$ cones $\PSD_d^k, \PCP_d^k$ and $\TCP_d^k$ for $k\geq 3$. Although no attempt to analyse these cones in full generality is made in this paper, we do provide simple necessary conditions for membership in these cones using the maps introduced in Eq.~\eqref{eq:M-k}.

\begin{proposition}
Let $B\in \PSD_d^k$. Then, $M_{k-1}(B)$ is positive semi-definite.
\end{proposition}

\begin{proof}
    Assume that the vectors $\{\ket{b_n}\}_{n\in I}$ form the rank one decomposition of $B\in \PSD_d^k$ as in Definition~\ref{def:psd-2}. Let us, for convenience, define $\operatorname{supp}_n$ as an index set of size $k$ which contains the actual (maybe smaller) support of the vector $\ket{b_n}\in \C{d}:\,\, \operatorname{supp}\ket{b_n}\subseteq \operatorname{supp}_n$ for each $n\in I$. Let $\ket{e}\in \mathbb{R}^d$ be the all ones vector. Then,
    \begin{align*}
        \langle e |M_{k-1}(B)| e\rangle &= \sum_{i=1}^d (k-1)B_{ii} - \sum_{1\leq i\neq j \leq d} |B_{ij}| \\ 
        &= \sum_{i=1}^d (k-1) \sum_{n\in I} |b_{n,i}|^2 - \sum_{i\neq j} \left| \sum_{n\in I} b_{n,i}\overbar{b_{n,j}}    \right|  \\
        &\geq \sum_{n\in I} (k-1)\sum_{i=1}^d |b_{n,i}|^2 - \sum_{n\in I}\sum_{i\neq j} |b_{n,i}b_{n,j}| \\
        &\geq \sum_{n\in I} \left(\sum_{ \substack{i,j\,\in\, \operatorname{supp}_n \\ i<j }} |b_{n,i}|^2 + |b_{n,j}|^2 \right) - \sum_{n\in I} \left(\sum_{ \substack{i,j\,\in\, \operatorname{supp}_n \\ i<j }} 2|b_{n,i}b_{n,j}|\right) \\
        &= \sum_{n\in I} \left( \sum_{ \substack{i,j\,\in\, \operatorname{supp}_n \\ i<j }} |b_{n,i}|^2 + |b_{n,j}|^2 - 2|b_{n,i} b_{n,j}|    \right) \geq 0
    \end{align*}
    Now, for any positive vector $\ket{\psi}\in \mathbb{R}^d_{+}$, we have $\langle \psi | M_{k-1}(B) | \psi \rangle = \langle e | D_{\psi} M_{k-1}(B) D_{\psi} | e \rangle  \geq 0$, where $D_{\psi}=\operatorname{diag}\ket{\psi}$. Moreover, if we assume (without loss of generality) that $M_{k-1}(B)$ is irreducible and write $M_{k-1}(B) = \alpha\mathbb{I}_d - P$ for $\alpha\geq 0$ and an irreducible $P\in \EWP_d$,  we can choose $\ket{\psi}$ to be the (normalized) Perron eigenvector of $P$ so that $ \langle \psi | M_{k-1}(B) | \psi \rangle = \alpha - \langle \psi | P | \psi \rangle = \alpha - \rho(P) \geq 0$, where $\rho(P)$ is the spectral radius of $P$. This shows that $M_{k-1}(B)$ is positive semi-definite.
\end{proof}

\begin{proposition}
Let $(A,B,C)\in \TCP_d^k$. Then $M_{k-1}(B)$ and $M_{k-1}(C)$ are positive semi-definite.
\end{proposition}

It is perhaps worthwhile to point out that the above necessary condition for membership in $\TCP_d^2$ has been utilized in \cite{Singh2020entanglement} to unearth a new kind of entanglement in bipartite quantum states which likes to hide behind peculiar distributions of zeros on the states' diagonals.

\section{The \texorpdfstring{PPT$^2$}{PPT squared} conjecture for diagonal unitary invariant maps} \label{sec:PPTsq}
In this section, we investigate the validity of the PPT$^2$ conjecture for diagonal unitary covariant maps in $\T{d}$. In particular, by exploiting the composition rule for these maps, along with the properties of factor width 2 pairwise completely positive matrices, we show that the PPT squared conjecture holds for linear maps in $\DUC_d$ and $\CDUC_d$. 

To begin with, let us consider two particular composition rules. 

\begin{definition}\label{def:DUC-CDUC-composition}
    On $\MLDUI{d}$, define bilinear compositions $\circ_1$ and $\circ_2$ as follows:
\begin{alignat}{2}
\circ_1 : \qquad \MLDUI{d} \,\,\, &\times \,\,\, \MLDUI{d} &&\rightarrow \MLDUI{d} \nonumber \\
\{ (A_1,B_1) &, (A_2,B_2) \} &&\mapsto (\mathfrak{A},\mathfrak{B}) = (A_1 A_2, B_1 \odot B_2^\top + \operatorname{diag}(A_1 A_2 - B_1\odot B_2)) \nonumber
\end{alignat}
\begin{alignat}{2}
\circ_2 : \qquad \MLDUI{d} \,\,\, &\times \,\,\, \MLDUI{d} &&\rightarrow \MLDUI{d} \nonumber \\
\{ (A_1,B_1) &, (A_2,B_2) \} &&\mapsto (\mathfrak{A},\mathfrak{B}) = (A_1 A_2, B_1\odot B_2 + \operatorname{diag}(A_1 A_2 - B_1\odot B_2)) \nonumber
\end{alignat}
\end{definition}

\begin{remark} \label{remark:DUC-CDUC-composition}
It is obvious from the above Definition that 
$$(A_1, B_1)\circ_1 (A_2, B_2) = (A_1, B_1)\circ_2 (A_2, B_2^\top) \qquad \forall (A_1, B_1), (A_2, B_2) \in \MLDUI{d}. $$
\end{remark}

Next, we state and prove an important proposition, which connects the above composition rules on matrix pairs to the operations of map composition in $\DUC_d$ and $\CDUC_d$. But first, we need familiarity with the notion of \emph{stability} under composition.
\begin{definition}
    A set $K\subseteq \T{d}$ is said to be stable under composition if $\Phi_1\circ \Phi_2 \in K$ for all $\Phi_{1},\Phi_2 \in K$.
\end{definition}

\begin{proposition} \label{prop:DUC-CDUC-composition}
$\CDUC_d \subset \T{d}$ is 
stable under composition, but $\DUC_d \subset \T{d}$ is not. Moreover, for pairs $(A_1,B_1), (A_2,B_2)\in \MLDUI{d}$, the following composition rules hold: 
$$\Phi^{(i)}_{(A_1,B_1)} \circ \Phi^{(j)}_{(A_2,B_2)} = \begin{cases}
    \, \Phi^{(2)}_{(\mathfrak{A},\mathfrak{B})},  &\text{where } (\mathfrak{A},\mathfrak{B}) = (A_1,B_1)\circ_i (A_2,B_2) \quad \text{if } i=j \\
    \, \Phi^{(1)}_{(\mathfrak{A},\mathfrak{B})},  &\text{where } (\mathfrak{A},\mathfrak{B}) = (A_1,B_1)\circ_i (A_2,B_2) \quad \text{if } i\neq j. 
    \end{cases} $$
\end{proposition}

\begin{proof}
    The stability results follow directly from Definition~\ref{def:DUC-CDUC-DOC}. It is also trivial to check that if $\Phi_1\in \DUC_d$ and $\Phi_2\in \CDUC_d$ (or vice-versa), then $\Phi_1\circ \Phi_2 \in \DUC_d$, since $\forall \, U\in \mathcal{DU}_d$ and $Z\in \M{d}$, the following equation holds: $[\Phi_1\circ \Phi_2](UZU^*) = \Phi_1 [U \Phi_2(Z) U^*] = U^* [\Phi_1\circ \Phi_2 (Z)] U$.
    
    Let us now show one of the composition rules, say the one for two CDUC maps, leaving the other three to the reader. Consider thus two CDUC maps acting as follows: 
    $$\Phi^{(2)}_{(A_i,B_i)}(X) = \operatorname{diag}(A_i \ket{\operatorname{diag}X}) + \widetilde{B_i}\odot X.$$
    We have
    \begin{align*}
        \Phi^{(2)}_{(A_1, B_1)} \circ \Phi^{(2)}_{(A_2, B_2)}(X) &=  \Phi^{(2)}_{(A_1, B_1)} \Big(\underbrace{ \operatorname{diag}(A_2 \ket{\operatorname{diag}X}) + \widetilde{B_2}\odot X }_{Y}\Big) \\
        &= \operatorname{diag}(A_1 \ket{\operatorname{diag}Y}) + \widetilde{B_1}\odot Y\\
        &= \operatorname{diag}(\mathfrak A \ket{\operatorname{diag}X}) + \widetilde{\mathfrak B}\odot X,
    \end{align*}
    where $(\mathfrak A, \mathfrak B)$ are given precisely by the composition rule $\circ_2$ from Definition \ref{def:DUC-CDUC-composition}. Indeed, for the diagonal part, since $\tilde B_2$ has zero diagonal, we have $\operatorname{diag}(Y) = A_2 \ket{\operatorname{diag}X}$, hence 
    $$\operatorname{diag}(A_1 \ket{\operatorname{diag}Y}) = \operatorname{diag}(A_1 A_2\ket{\operatorname{diag}X}).$$
    proving the claim for $\mathfrak A$. Similarly, in $\widetilde B_1 \odot Y$, only the off-diagonal entries of $Y$ matter, and those are given by $\widetilde B_2 \odot X$. Hence, 
    $$\widetilde B_1 \odot Y = \widetilde B_1 \odot \left( \widetilde B_2 \odot X \right) = (\widetilde{B_1 \odot B_2})\odot X,$$
    proving that $\widetilde{\mathfrak B} = \widetilde{B_1 \odot B_2}$ and finishing the proof.
\end{proof}

We are now finally ready to prove the main result of our paper.

\begin{theorem}\label{thm:PPT2-C-DUC}
Consider matrix pairs $(A,B), (C,D) \in \MLDUI{d}$ such that the corresponding (C)DUC linear maps $\Phi^{(i)}_{(A,B)}, \Phi^{(i)}_{(C,D)} \in \mathcal{T}_d(\mathbb{C})$ are \emph{PPT} for $i=1,2$. Then, the compositions $\Phi^{(i)}_{(A,B)} \circ \Phi^{(j)}_{(C,D)}$ are entanglement breaking for $i,j=1,2$. 
\end{theorem}

\begin{proof} 
First of all, invoke Proposition~\ref{prop:DUC-CDUC-composition} to deduce that 
$$ \Phi^{(i)}_{(A,B)} \circ \Phi^{(j)}_{(C,D)} = \begin{cases}
\, \Phi^{(2)}_{(\mathfrak{A},\mathfrak{B})},  &\text{where }            (\mathfrak{A},\mathfrak{B}) = (A,B)\circ_i (C,D) \quad \text{if } i=j \\
 \, \Phi^{(1)}_{(\mathfrak{A},\mathfrak{B})},  &\text{where } (\mathfrak{A},\mathfrak{B}) = (A,B)\circ_i (C,D) \quad \text{if } i\neq j
\end{cases}
    $$
Hence to prove the Theorem, it suffices to show that the matrix pairs $(\mathfrak{A}',\mathfrak{B}') = (A,B)\circ_1 (C,D)$ and $(\mathfrak{A},\mathfrak{B}) = (A,B)\circ_2 (C,D)$ are PCP, see Proposition~\ref{prop:PPT-ent-DUC/CDUC}. In what follows, we show that the latter pair $(\mathfrak{A},\mathfrak{B})$ is PCP, and the former case will follow from Remark~\ref{remark:DUC-CDUC-composition}. With this end in sight, let us analyze the structure of $\mathfrak{B}=\operatorname{diag}AC + \widetilde{B}\odot \widetilde{D}$ in some detail. Since all the involved maps in this Theorem are PPT, the compositions are guaranteed to be PPT as well, which implies that $\mathfrak{B}\in \PSD_d$. Let us initially restrict ourselves to the case when $d=2$. Here, we have 
    \begin{equation*}
        \mathfrak{B} = \left(  \begin{array}{cc}
        A_{11}C_{11} + A_{12}C_{21} & B_{12}D_{12}  \\
        B_{21}D_{21} &  A_{21}C_{12} + A_{22}C_{22} 
    \end{array}  \right) = \left(  \begin{array}{cc}
        A_{11}C_{11} & 0   \\
        0  & A_{22}C_{22} 
    \end{array}  \right) + \left(  \begin{array}{cc}
         A_{12}C_{21} &  B_{12}D_{12}  \\
         B_{21}D_{21} &  A_{21}C_{12}
    \end{array}  \right)
    \end{equation*}
Since $A,C\in \EWP_d$ and $A_{12}A_{21}C_{12}C_{21} \geq | B_{12}D_{12} |^2$ (because the maps $\Phi^{(i)}_{(A,B)}$ and $\Phi^{(i)}_{(C,D)}$ are PPT), it is clear that $\mathfrak{B}\in \PSD_2^2$. Similar splitting can be performed in the $d=3$ case:
\begin{align*}
\mathfrak{B} &= \left(  \begin{array}{ccc}
    A_{11}C_{11} + A_{12}C_{21} + A_{13}C_{31} & B_{12}D_{12} & B_{13}D_{13} \\
    B_{21}D_{21} &  A_{21}C_{12} + A_{22}C_{22} + A_{23}C_{32} & B_{23}D_{23} \\
    B_{31}D_{31}  &  B_{32}D_{32}  & A_{31}C_{13} + A_{32}C_{23} + A_{33}C_{33}
\end{array}  \right) \nonumber \\
    &= \operatorname{diag}(A\odot C) + \left(  \begin{array}{ccc}
        A_{12}C_{21} &  B_{12}D_{12}   & 0 \\
        B_{21}D_{21}  &  A_{21}C_{12} & 0 \\
        0 & 0 & 0
    \end{array}  \right) + \left(  \begin{array}{ccc}
        A_{13}C_{31} & 0 &  B_{13}D_{13}  \\
        0 &  0 & 0 \\
         B_{31}D_{31}  & 0 & A_{31}C_{13}
    \end{array}  \right) +  \left(  \begin{array}{ccc}
        0 & 0 & 0 \\
        0 & A_{23}C_{32} &  B_{23}D_{23}   \\
        0 &   B_{32}D_{32}  & A_{32}C_{23}
    \end{array}  \right) 
\end{align*}
which shows that $\mathfrak{B}\in \PSD_3^2$. It should now be apparent that the above splittings are nothing but special cases of a more general decomposition, which holds for $\mathfrak{B}\in \PSD_d$ for arbitrary $d\in \mathbb{N}$:
\begin{equation*}
    \mathfrak{B} = \operatorname{diag}(A\odot C) + \sum_{1\leq i<j \leq d} \left(\begin{array}{ c c }
        A_{ij}C_{ji} &  B_{ij}D_{ij} \\
        B_{ji}D_{ji}  &  A_{ji}C_{ij}
    \end{array} \right)_{i,j\in [d]}
\end{equation*}
Observe that we used the notation from Lemma~\ref{lemma:B-DD} above. The PPT constraints $A,C\in \EWP_d$ and $A_{ij}A_{ji}C_{ij}C_{ji}\geq |B_{ij}D_{ij}|^2 \,\, \forall i,j$ imply that all the matrices in the above decomposition are positive semi-definite, which shows that $\mathfrak{B}$ has factor width 2. A swift application of Theorem~\ref{theorem:PCP-2} then shows that $(\mathfrak{A},\mathfrak{B})\in \PCP_d^2$.
\end{proof}    
    
We should emphasize here that the above theorem contains a \emph{stronger} form of the PPT$^2$ conjecture for (C)DUC maps: we show that the composition of two PPT (C)DUC maps corresponds to a matrix pair $(\mathfrak{A},\mathfrak{B})\in \PCP_d^2$, where $\PCP_d^2$ is a strict subset of $\PCP_d$, for $d \geq 3$, see Remark \ref{rem:strict-inclusion}. 

Let the vectors $\{\ket{v_n},\ket{w_n} \}_{n\in I}\subset \C{d}$ form the PCP decomposition of $(\mathfrak{A},\mathfrak{B})$ as in Definition~\ref{def:pcp/tcp-2}. By defining $\mathfrak{A}_n = \ketbra{v_n\odot \overbar{v_n}}{w_n\odot \overbar{w_n}}$ and $\mathfrak{B}_n=\ketbra{v_n\odot w_n}{v_n\odot w_n}$, we see that for $i=1,2$, the choi matrix $J(\Phi^{(i)}_{(\mathfrak{A},\mathfrak{B})})$ splits up into a sum of PPT choi matrices $J(\Phi^{(i)}_{(\mathfrak{A}_n,\mathfrak{B}_n)})$, each having support on a $2\otimes 2$ subsystem (barring some diagonal entries). Hence, we can write
\begin{equation} \label{eq:compose-2}
 J(\Phi^{(i)}_{(\mathfrak{A},\mathfrak{B})}) = \sum_{n\in I}J(\Phi^{(i)}_{(\mathfrak{A}_n,\mathfrak{B}_n)}) \implies 
 \Phi^{(i)}_{(\mathfrak{A},\mathfrak{B})} = \sum_{n\in I} \Phi^{(i)}_{(\mathfrak{A}_n,\mathfrak{B}_n)},
\end{equation}
where, since each $\Phi^{(i)}_{(\mathfrak{A}_n,\mathfrak{B}_n)}$ is a PPT map acting on a qubit system, the sum $\Phi^{(i)}_{(\mathfrak{A},\mathfrak{B})}$ is entanglement breaking, for $i=1,2$.

\section{Perspectives and open questions}\label{sec:open}

Let us first emphasize that, in light of Proposition \ref{prop:PPT2}, the main result in Theorem \ref{thm:PPT2-C-DUC} can be interpreted as follows. Assume Alice, Bob, and Charlie share a state $\rho \otimes \sigma$ as in Figure~\ref{fig:swap}, with $\rho, \sigma\in \M{d}\otimes \M{d}$ being PPT and (conjugate) local diagonal unitary invariant in the sense of Remark~\ref{remark:DOC-LDOI}. Alice acts on her system with a (C)DUC map $\Phi_A$; similarly, Bob acts on his system with a (C)DUC map $\Psi_B$. Charlie then postselects on his bipartite system, the maximally entangled state $\ket{e}$. Then, Theorem \ref{thm:PPT2-C-DUC} says that the resulting state of Alice and Bob is separable. Note however that the local actions of Alice, Bob, and Charlie are restricted; a general result, in the sense of point (3) of Proposition \ref{prop:PPT2} does not hold in our diagonal-covariant setting. 

The main question left open in this work is the PPT$^2$ question for DOC maps. We conjecture that answer is the same as for (C)DUC maps. 
\begin{conjecture} \label{conjecture}
The composition of two arbitrary \emph{PPT} maps in $\DOC_d$ is entanglement breaking. 
\end{conjecture}
We recall (see \cite[Lemma 9.3]{Singh2020diagonal}) the composition rule for DOC maps. Given two triples of matrices $(A_1,B_1,C_1), (A_2,B_2,C_2)\in \MLDOI{d}$, one has 
$$\Phi^{(3)}_{(A_1,B_1,C_1)} \circ \Phi^{(3)}_{(A_2,B_2,C_2)} = \Phi^{(3)}_{(\mathfrak{A},\mathfrak{B},\mathfrak{C})},$$ where 
\begin{align*}
    \mathfrak{A} &= A_1 A_2 \\ 
    \mathfrak{B} &= B_1 \odot B_2 + C_1\odot C_2^\top + \operatorname{diag}(A_1 A_2 - 2A_1\odot A_2) \\
    \mathfrak{C} &= B_1\odot C_2 + C_1\odot B_2^\top + \operatorname{diag}(A_1 A_2 - 2A_1\odot A_2).
\end{align*}

Hence, in terms of matrix triples, Conjecture~\ref{conjecture} posits that $(\mathfrak{A},\mathfrak{B},\mathfrak{C})\in \TCP_d$ for arbitrary matrix triples $(A_i,B_i,C_i)\in \MLDOI{d}$ such that $\Phi^{(3)}_{(A_i,B_i,C_i)}$ is PPT, for $i=1,2$. 

We would like to point out that even if the conjecture above holds in full generality, we do not have evidence for a stronger conclusion, as is the case with (C)DUC maps. Indeed, we have numerical evidence for the existence of PPT triples $(A_1, B_1, C_1), (A_2, B_2, C_2)$ (see Proposition \ref{prop:PPT-ent-DOC}) such that their compositions $(\mathfrak A, \mathfrak B, \mathfrak C)$ have factor width $>2$. More precisely, we have found PPT matrix pairs $(A,B) \in \MLDUI{d}$ such that $(\mathfrak{A},\mathfrak{B})\notin \PCP_d^2$, where $(\mathfrak{A},\mathfrak{B},\mathfrak{B})$ is obtained by composing the triple $(A,B,B)$ with itself in the above fashion. Hence, in order to prove the PPT$^2$ conjecture for the more general class of DOC maps, one likely requires stronger criterion for separability, in terms of sufficient conditions for membership in both the $\PCP_d$ and $\TCP_d$ cones. 

\bigskip

\noindent \textit{Acknowledgements.} We thank Alexander M\"uller-Hermes for valuable correspondence on the PPT$^2$ conjecture, and Salman Beigi and Ludovico Lami for directing our attention to the reference \cite{ringbauer2018certification}, where the notion of factor width has been employed to study coherence of quantum states. 

\bibliographystyle{alpha}
\bibliography{references}

\bigskip
\hrule
\bigskip
\end{document}